\documentclass[runningheads]{llncs}
\usepackage[T1]{fontenc}
\usepackage{graphicx}
\usepackage{amsmath}

\spnewtheorem{observation}[theorem]{Observation}{\bfseries}{\itshape}

\begin{document}
\title{Colored Interaction-Profile Realization: Complexity of Matching-Match on Spiders}

\titlerunning{Complexity of Matching-Match on Spiders}

\author{Ilie Dumitru\inst{1} \and
Adrian Micl\u au\c s\inst{1} \and
Alexandru Popa\inst{1, 2}}
\authorrunning{Dumitru et al.}

\institute{Department of Computer Science, University of Bucharest, Str. Academiei 14, Bucharest, 010014, Romania \and
National Institute for Research and Development in Informatics, Bulevardul Mareșal Alexandru Averescu 8-10, Bucharest, 011555, Romania}
\maketitle              % typeset the header of the contribution
\begin{abstract}
Network motifs and colored local interaction patterns provide a useful way to describe the structure of complex networks. Motivated by an inverse realization perspective, we study the problem of assigning colors to the vertices of a fixed graph so that its edges realize a prescribed multiset of colored pairwise interactions. This problem is formalized by the Matching-Match Puzzle, introduced by Iburi and Uehara.

We investigate how its computational complexity depends on the number of colors and on the structure of the host graph. We first prove that Matching-Match is NP-complete with only two colors, even when no vertex is precolored and the graph has maximum degree three. We then focus on spiders. We show that the problem is W[1]-hard parameterized by the number of colors even on spiders with only the body precolored. In contrast, for spiders whose legs have length at most two, we give a fixed-parameter tractable algorithm parameterized by the number of colors, allowing arbitrary precoloring. Finally, we prove NP-completeness for spiders whose legs all have length exactly three when precoloring is allowed.

\end{abstract}

\section{Introduction}

\subsection*{Motivation}

Complex networks are commonly studied through small recurring patterns of interactions, known as \emph{network motifs}~\cite{Milo2002NetworkMotifs}. Such patterns have been observed in biological, ecological, technological, and information-processing networks and have been proposed as basic structural building blocks of complex systems~\cite{Milo2002NetworkMotifs,Alon2007NetworkMotifs}. When the vertices of a network represent entities with different functions or types, this idea can be refined through \emph{colored motifs}, where colors encode the roles of the participating vertices~\cite{Adami2011ColoredMotifs,Qian2011ColoredMotifs}.

Motif-based approaches usually start from an already labeled network and ask which local patterns occur in it. This naturally suggests an inverse question: given the topology of a network and a prescribed collection of colored interactions, can the vertices be labeled so that these interactions are realized simultaneously? We study this type of realization problem through the \textsc{Matching-Match Puzzle}, introduced by Iburi and Uehara~\cite{IburiUehara2024}. Although the problem originates from a physical puzzle, it has a natural interpretation as the exact realization of a prescribed colored interaction profile on a fixed graph.

\subsection*{Informal Problem Definition}

An instance of \textsc{Matching-Match} consists of a graph whose vertices may be partially precolored and a collection of colored sticks, one for each graph edge. Every stick has a color assigned to each of its two endpoints. The objective is to place the sticks on the graph edges so that the prescribed colors of precolored vertices are respected and all stick endpoints incident with the same initially uncolored vertex have a common color.

Since the number of sticks equals the number of graph edges, every solution assigns exactly one stick to each edge. Equivalently, each stick specifies an unordered pair of colors, and the problem asks whether the graph can be vertex-colored so that the multiset of color pairs induced by its edges is exactly the multiset prescribed by the sticks. From this perspective, the sticks describe the desired colored interaction profile, while the vertex coloring realizing this profile is unknown.

\subsection*{Previous and Related Work}

Network motifs were introduced as statistically significant recurring subgraphs of complex networks by Milo et al.~\cite{Milo2002NetworkMotifs}. Colored variants were subsequently considered to incorporate functional information associated with the network vertices~\cite{Adami2011ColoredMotifs,Qian2011ColoredMotifs}. A related but different algorithmic problem is \textsc{Graph Motif}, motivated by biological networks, where one searches for a connected subgraph whose vertex colors realize a prescribed multiset~\cite{Fellows2011GraphMotif,Betzler2011GraphMotif}. Graph Motif has received considerable attention from the viewpoint of parameterized and structural complexity~\cite{BonnetSikora2017GraphMotif}. In contrast, Matching-Match keeps the entire host graph fixed and asks whether its unknown vertex colors can realize a prescribed collection of colored edge interactions.

The computational complexity of puzzles and games has also been extensively investigated as a source of natural combinatorial problems~\cite{HearnDemaine2009,Uehara2023Puzzles}. Within this line of research, Iburi and Uehara introduced the \textsc{Matching-Match Puzzle} and initiated its complexity study~\cite{IburiUehara2024}. They proved NP-completeness even on highly restricted graph classes, including paths and spiders, while obtaining polynomial-time algorithms for complete graphs and stars. Of particular relevance to our work, they showed that completely uncolored spiders whose legs have length at most two can be solved in polynomial time. They also identified the number of colors, additional precoloring on short spiders, and the maximum length of spider legs as natural directions for further investigation~\cite{IburiUehara2024}.

\subsection*{Our Results}

We continue the complexity study of Matching-Match, focusing on the interaction between the number of colors and the structure of spiders.

\paragraph{Hardness with two colors}

We first prove that \textsc{Matching-Match} is NP-complete when only two colors are available, even if no graph vertex is initially precolored and the host graph has maximum degree three. Since the one-color case is trivial, this establishes hardness for the minimum nontrivial number of colors. In particular, the problem is para-NP-hard when parameterized by the number of colors on general graphs.

\paragraph{Parameterized hardness on spiders}

We then show that restricting the graph to a spider does not make the number of colors a tractable parameter. We prove that \textsc{Matching-Match} is W[1]-hard parameterized by the number of colors even when the host graph is a spider and only its body is precolored. Thus, the problem remains parameterized intractable on a subclass of trees, despite the fact that spiders have treewidth one.

\paragraph{Fixed-parameter tractability for short spiders}

The hardness above relies on spiders with arbitrarily long legs. In contrast, we prove that \textsc{Matching-Match} is fixed-parameter tractable parameterized by the number of colors when every leg has length at most two, even under arbitrary precoloring. Our algorithm groups legs according to their precoloring patterns and reduces feasibility to an integer linear program whose dimension depends only on the number of colors. This gives parameterized progress on the precoloring problem left open by Iburi and Uehara~\cite{IburiUehara2024}, although its classical polynomial-versus-NP-complete complexity remains open when the number of colors is unbounded.

\paragraph{Hardness at leg length three}

Finally, we consider the first leg length beyond the previously known tractable case. We prove that \textsc{Matching-Match} is NP-complete on spiders in which every leg has length exactly three when precoloring is allowed. Together with the polynomial-time result of Iburi and Uehara for completely uncolored spiders of leg length at most two~\cite{IburiUehara2024}, this substantially narrows the complexity boundary with respect to leg length. The complexity of completely uncolored spiders whose legs have length exactly three remains open.

\section{Preliminaries}
\label{sec:preliminaries}

Throughout the paper, graphs are finite, simple, and undirected. For a graph $G=(V,E)$, we write $n=|V|$ and $m=|E|$. For a vertex $v\in V$, its degree is denoted by $\deg_G(v)$, and $\Delta(G)$ denotes the maximum degree of $G$.

A path of length $\ell$ is a sequence of vertices $v_0,v_1,\ldots,v_\ell$ such that $v_iv_{i+1}\in E$ for every $0\le i<\ell$. A tree is a connected acyclic graph. A \emph{spider} is a tree containing a unique vertex of degree greater than two. This vertex is called the \emph{body} of the spider. Every maximal path starting at the body and ending at a leaf is called a \emph{leg}; the length of a leg is its number of edges. A star is a spider in which every leg has length one.

We now recall the \textsc{Matching-Match Puzzle} introduced by Iburi and Uehara~\cite{IburiUehara2024}. Let $C=\{1,\ldots,c\}$ be a set of colors. An instance consists of a graph $G=(V,E)$, a set $S$ of $m=|E|$ sticks, a partial coloring $\mathcal C_0:V\rightarrow \{0,1,\ldots,c\}$ of the graph vertices, and a coloring $\mathcal C_1$ of all stick endpoints with colors from $C$. The value $\mathcal C_0(v)=0$ means that $v$ is initially uncolored.

Each stick $s\in S$ has two colored endpoints. Since the graph is undirected, the orientation of a stick on an edge is irrelevant. We therefore define the \emph{type} of a stick to be the unordered pair $\{i,j\}$ consisting of the colors of its endpoints. For $1\le i\le j\le c$, let $s_{ij}$ denote the number of sticks of type $\{i,j\}$.

A mapping of the stick endpoints to graph vertices is \emph{feasible} if every stick is mapped to an edge of $G$, every graph edge receives a stick, every precolored graph vertex agrees with the colors of all stick endpoints mapped to it, and all stick endpoints mapped to the same initially uncolored graph vertex have the same color. Since $|S|=|E|$, every feasible mapping induces a bijection between the sticks and the graph edges. Hence it is often more convenient to use an equivalent coloring formulation.

\begin{observation}
\label{obs:coloring-formulation}
An instance of \textsc{FMMP} is feasible if and only if there exists a coloring
$\phi:V\rightarrow C$ extending $\mathcal C_0$ such that, for every
$1\leq i\leq j\leq c$, the number of edges $uv\in E$ satisfying
$\{\phi(u),\phi(v)\}=\{i,j\}$ is exactly $s_{ij}$.
\end{observation}

Indeed, every feasible mapping assigns a unique color to each initially uncolored vertex and therefore induces such a coloring $\phi$. Conversely, if such a coloring exists, then for every type $\{i,j\}$ the number of graph edges of that type equals the number of sticks of that type, so the sticks can be assigned bijectively to graph edges type by type. We use \textsc{FMMP} to denote the corresponding decision problem.

\begin{problem}[Feasible Matching-Match Puzzle (\textsc{FMMP})]
Given a graph $G=(V,E)$, a set $S$ of $|E|$ colored sticks using colors from $C=\{1,\ldots,c\}$, and a partial coloring $\mathcal C_0:V \rightarrow \{0, 1,\ldots,c\}$, determine whether there exists a feasible mapping. Equivalently, determine whether there exists a coloring $\phi : V \rightarrow C$ extending $\mathcal C_0$ whose edge-type multiplicities coincide with the stick-type multiplicities.
\end{problem}

An instance is called \emph{uncolored} if $\mathcal C_0(v)=0$ for every $v\in V$. When the body of a spider is the only precolored vertex, all vertices on its legs are initially uncolored.

\section{Hardness with Two Colors}
\label{sec:two-colors}

We begin by studying the influence of the number of colors. Iburi and Uehara asked whether bounding the number of colors may lead to tractability~\cite{IburiUehara2024}. We show that this is not the case on general graphs: two colors, which is the minimum nontrivial number, already suffice for NP-completeness. Moreover, the hardness holds without any precolored vertices and on graphs of maximum degree three.

\begin{theorem}
\label{thm:two-colors}
\textsc{FMMP} is NP-complete even when $c=2$, no vertex is precolored, and the input graph has maximum degree at most $3$.
\end{theorem}

\begin{proof}
Membership in NP follows immediately, since a mapping of the sticks to the graph edges can be verified in polynomial time.

For NP-hardness, we reduce from \textsc{Max-Cut} on cubic graphs. Alimonti and Kann showed that \textsc{Max-Cut} is APX-complete even on cubic graphs~\cite{AlimontiKann2000}. In particular, its decision version is NP-complete.

Let $(H,k)$ be an instance of \textsc{Max-Cut}, where $H=(W,F)$ is a cubic graph and $m=|F|$. We ask whether $H$ contains a cut of size at least $k$.

We construct an instance $I=(G,S,2,\mathcal C_0,\mathcal C_1)$ of \textsc{FMMP}. The graph $G$ consists of two vertex-disjoint copies $H_1$ and $H_2$ of $H$, together with $2(m-k)$ independent edges. All vertices of $G$ are initially uncolored. Hence $|E(G)|=2m+2(m-k)=4m-2k$.

We use the two colors $1$ and $2$. The set of sticks consists of $m-k$ sticks of type $\{1,1\}$, $2m$ sticks of type $\{1,2\}$, and $m-k$ sticks of type $\{2,2\}$. Thus $|S|=(m-k)+2m+(m-k)=4m-2k=|E(G)|$.

The construction can be performed in polynomial time. Moreover, since $H$ is cubic and every additional component is a single edge, $\Delta(G)\leq 3$.

We prove that $H$ has a cut of size at least $k$ if and only if the constructed instance is feasible.

For the forward direction, suppose that $H$ has a cut $(A,W\setminus A)$ of size $t\geq k$. Color the vertices of $H_1$ by assigning color $1$ to the vertices corresponding to $A$ and color $2$ to the remaining vertices. Color $H_2$ in the opposite way: vertices corresponding to $A$ receive color $2$, and all other vertices receive color $1$.

In each copy, exactly $t$ edges cross the cut and therefore have type $\{1, 2\}$. Across the two copies, there are consequently $2t$ edges of type $\{1, 2\}$.

Now consider the $m-t$ edges of $H$ that do not cross the cut. If such an edge has type $\{1,1\}$ in $H_1$, then the corresponding edge has type $\{2,2\}$ in $H_2$, and conversely. Therefore, over $H_1\cup H_2$, there are exactly $m-t$ edges of type $\{1,1\}$ and exactly $m-t$ edges of type $\{2,2\}$.

After assigning sticks of the corresponding types to all edges of $H_1$ and $H_2$, the remaining sticks are $t-k$ sticks of type $\{1,1\}$, $t-k$ sticks of type $\{2,2\}$, and $2(m-t)$ sticks of type $\{1, 2\}$. Their total number is $2(t-k)+2(m-t)=2(m-k)$, which is exactly the number of independent buffer edges.

Since the endpoints of every buffer edge are uncolored and have no other incident edges, we may color them according to the type of the stick assigned to that edge. Hence all remaining sticks can be placed on the buffer edges, and the constructed instance is feasible.

For the reverse direction, suppose that the constructed \textsc{FMMP} instance is feasible. Since there are exactly $2m$ sticks of type $\{1, 2\}$, exactly $2m$ graph edges must have endpoints of different colors.

At most $2(m-k)$ of these bichromatic edges can belong to the independent buffer edges. Therefore, at least $2m-2(m-k)=2k$ bichromatic edges belong to $H_1\cup H_2$.

Consequently, at least one of the two copies $H_1$ and $H_2$ contains at least $k$ bichromatic edges. The two colors assigned to the vertices of that copy define a cut of the original graph $H$, and the bichromatic edges are precisely the edges crossing this cut. Hence $H$ has a cut of size at least $k$.
\end{proof}

In particular, \textsc{FMMP} is para-NP-hard when parameterized by the number of colors on general graphs.

\section{Parameterized Hardness on Spiders}
\label{sec:spider-w1-hardness}

Theorem~\ref{thm:two-colors} shows that bounding the number of colors alone does not yield tractability on general graphs. We now ask whether the situation changes when the host graph is restricted to a very simple tree class. We show that it does not: \textsc{FMMP} remains W[1]-hard parameterized by the number of colors even on spiders. Moreover, the reduction requires only the body of the spider to be precolored.

We reduce from \textsc{Unary Bin Packing}. In this problem, we are given positive integer item sizes $a_1,\ldots,a_n$, encoded in unary, together with $k$ bins of capacity $B$, and we ask whether the items can be partitioned among the bins so that the total size assigned to every bin is at most $B$. Jansen et al.~\cite{JansenKratschMarxSchlotter2013} proved that \textsc{Unary Bin Packing} is W[1]-hard parameterized by the number $k$ of bins.

We may assume that the instance is \emph{tight}, that is, $\displaystyle\sum_{i=1}^{n}a_i=kB$. Indeed, if the total size is larger than $kB$, the instance is immediately a no-instance. Otherwise, we may add $kB-\displaystyle\sum_{i=1}^{n}a_i$ unit-size items. The resulting instance is feasible if and only if the original instance is feasible: any unused capacity of a feasible packing can be filled with the added unit items, while deleting the added items from a feasible packing of the new instance gives a feasible packing of the original one. Since the numbers are encoded in unary, this transformation has polynomial size.

\begin{theorem}
\label{thm:spider-w1-hard}
\textsc{FMMP} is W[1]-hard parameterized by the number of colors $c$, even when the input graph is a spider and only its body is precolored.
\end{theorem}

\begin{proof}
Let $a_1,\ldots,a_n$, $k$, and $B$ be a tight instance of \textsc{Unary Bin Packing}, so $\displaystyle\sum_{i=1}^{n}a_i=kB$. We construct an instance of \textsc{FMMP} using the color set $C=\{1,\ldots,k+1\}$. Color $1$ will be used for the body of the spider, while colors $2,\ldots,k+1$ correspond to the $k$ bins.

We construct a spider with body $b$, which is precolored with color $1$. For every item $a_i$, we introduce one \emph{item leg} of length $a_i+1$. In addition, we introduce $n(k-1)$ \emph{buffer legs}, each of length $1$. All vertices other than the body are initially uncolored. The spider therefore has $n+n(k-1)=nk$ legs and $nk+\displaystyle\sum_{i=1}^{n}a_i=nk+kB$ edges.

We now define the sticks. For every $j\in \{1,\ldots,k\}$, we introduce $n$ sticks of type $\{1,j+1\}$ and $B$ sticks of type $\{j+1,j+1\}$. Hence, there are $nk$ sticks containing color $1$ and $kB$ monochromatic sticks. The total number of sticks is therefore $nk+kB$, which is exactly the number of edges of the constructed spider.

We claim that the constructed \textsc{FMMP} instance is feasible if and only if the \textsc{Unary Bin Packing} instance is feasible.

Suppose first that the items admit a feasible packing into the $k$ bins. Since the instance is tight, every bin has total size exactly $B$. Consider an item $a_i$ assigned to bin $j$. On the first edge of the corresponding item leg, place a stick of type $\{1,j+1\}$. On each of the remaining $a_i$ edges of that leg, place a stick of type $\{j+1,j+1\}$. Thus, the vertices of the leg other than the body receive color $j+1$.

Let $r_j$ be the number of items assigned to bin $j$. Since these items have total size $B$, the item legs assigned to bin $j$ consume exactly $B$ sticks of type $\{j+1,j+1\}$ and $r_j$ sticks of type $\{1,j+1\}$. There remain $n-r_j$ sticks of type $\{1,j+1\}$. Over all bins, the number of remaining sticks containing color $1$ is $\displaystyle\sum_{j=1}^{k}(n-r_j)=nk-n=n(k-1)$, exactly the number of buffer legs. We place one remaining stick on each buffer leg. Since their leaves are uncolored, every such placement is feasible. Hence the constructed instance is feasible.

Conversely, suppose that the constructed \textsc{FMMP} instance is feasible. The body is precolored with color $1$, so every edge incident with the body must receive a stick containing color $1$. There are exactly $nk$ edges incident with the body and exactly $nk$ sticks containing color $1$. Consequently, every stick containing color $1$ is used on an edge incident with the body.

Consider an item leg corresponding to an item $a_i$. Its first edge must therefore receive a stick of type $\{1,j+1\}$ for some $j\in{1,\ldots,k}$. The first vertex after the body consequently receives color $j+1$. Since all sticks containing color $1$ are already used on edges incident with the body, every remaining edge of the item leg must receive a monochromatic stick. The next edge can therefore only receive a stick of type $\{j+1,j+1\}$, which forces the next vertex to have color $j+1$. Repeating this argument along the leg shows that all $a_i$ remaining edges receive sticks of type $\{j+1,j+1\}$.

We assign item $a_i$ to bin $j$. For every $j$, there are exactly $B$ sticks of type $\{j+1,j+1\}$. Moreover, the total number of monochromatic sticks is $kB=\displaystyle\sum_{i=1}^{n}a_i$, which is exactly the total number of non-body edges on the item legs. Hence all monochromatic sticks are used on item legs, and for every $j$ the items assigned to bin $j$ have total size exactly $B$. They therefore define a feasible packing into the $k$ bins.

The construction is polynomial because the item sizes are encoded in unary. Finally, the number of colors in the constructed instance is $c=k+1$. Thus the parameter depends only on the number of bins of the source instance, and the construction is a parameterized reduction. Since \textsc{Unary Bin Packing} is W[1]-hard parameterized by $k$~\cite{JansenKratschMarxSchlotter2013}, the theorem follows.
\end{proof}

Theorem~\ref{thm:spider-w1-hard} shows that the simple tree structure of spiders is not sufficient to make the number of colors an FPT parameter. In particular, the hardness already holds on graphs of treewidth one. The reduction, however, uses item legs whose lengths depend on the item sizes. This motivates the bounded-leg case considered in the next section.

\section{Fixed-Parameter Tractability for Short Spiders}
\label{sec:short-spiders-fpt}

The W[1]-hardness result of Section~\ref{sec:spider-w1-hardness} relies on spiders whose legs may be arbitrarily long. Iburi and Uehara showed that \textsc{FMMP} is polynomial-time solvable on completely uncolored spiders whose legs have length at most $2$, and left open the complexity of this graph class when vertices may be precolored~\cite{IburiUehara2024}. We make parameterized progress on this open problem: under arbitrary precoloring, \textsc{FMMP} is fixed-parameter tractable with respect to the number of colors $c$.

Our algorithm groups legs according to their precoloring patterns and records only how many legs of each pattern receive each possible final coloring. Since every leg contains at most two vertices other than the body, the number of relevant patterns and color assignments depends only on $c$. We then express feasibility as an integer linear program whose number of variables is bounded by a function of $c$.

\begin{theorem}
\label{thm:short-spider-fpt}
\textsc{FMMP} on spiders whose legs have length at most $2$ is fixed-parameter tractable parameterized by the number of colors $c$, even under arbitrary precoloring.
\end{theorem}

\begin{proof}
Let $G$ be a spider with body $b$, and let $C=\{1,\ldots,c\}$. We first fix the final color $s\in C$ of the body. If $b$ is precolored, there is only one possible choice for $s$. Otherwise, we try all $c$ possibilities. It therefore suffices to solve the problem for a fixed body color $s$.

We use $0$ to denote an initially uncolored vertex. Since every leg has length at most $2$, there are two types of legs.

Consider first a leg of length $1$, consisting of the edge $bv$. Its precoloring pattern is determined by $\mathcal C_0(v)\in{0,\ldots,c}$. For every $p\in{0,\ldots,c}$, let $L_p$ denote the number of length-$1$ legs whose leaf is precolored $p$. We introduce an integer variable $x_{p,a}$ for every $p\in{0,\ldots,c}$ and $a\in C$. The variable $x_{p,a}$ denotes the number of such legs whose leaf receives final color $a$. We require $x_{p,a}=0$ whenever $p\neq 0$ and $p\neq a$, and for every $p$ we impose $ \displaystyle\sum_{a\in C}x_{p,a}=L_p$.

Now consider a leg of length $2$, written as $buv$, where $u$ is adjacent to the body and $v$ is the leaf. Its precoloring pattern is the pair $(\mathcal C_0(u),\mathcal C_0(v))$. For every $p,q\in \{0,\ldots,c\}$, let $L_\{p,q\}$ denote the number of length-$2$ legs having pattern $(p,q)$. We introduce an integer variable $y_{p,q,a,d}$ for every $p,q\in{0,\ldots,c}$ and $a,d\in C$. The variable $y_{p,q,a,d}$ denotes the number of these legs for which $u$ receives final color $a$ and $v$ receives final color $d$.

Again, precoloring is enforced by setting $y_{p,q,a,d}=0$ whenever $p\neq 0$ and $p\neq a$, or whenever $q\neq 0$ and $q\neq d$. For every pair $(p,q)$, we impose $\displaystyle\sum_{a,d\in C}y_{p,q,a,d}=L_\{p,q\}$.

It remains to ensure that the resulting coloring produces exactly the edge-type multiplicities prescribed by the sticks.

For every $1\leq i\leq j\leq c$, let $s_{ij}$ denote the number of sticks of type $\{i,j\}$. A length-$1$ leg whose leaf receives color $a$ contributes one edge of type ${s,a}$. A length-$2$ leg colored $s-a-d$ contributes one edge of type ${s,a}$ and one edge of type ${a,d}$.

For every unordered color pair $\{i,j\}$, we therefore impose

\[
\begin{aligned}
s_{ij}
&=
\displaystyle\sum_{p=0}^{c}\displaystyle\sum_{a\in C}
\mathbf{1}\!\left[\{s,a\}=\{i,j\}\right]x_{p,a}
\\
&\quad+
\displaystyle\sum_{p,q=0}^{c}\displaystyle\sum_{a,d\in C}
\mathbf{1}\!\left[\{s,a\}=\{i,j\}\right]y_{p,q,a,d}
\\
&\quad+
\displaystyle\sum_{p,q=0}^{c}\displaystyle\sum_{a,d\in C}
\mathbf{1}\!\left[\{a,d\}=\{i,j\}\right]y_{p,q,a,d}.
\end{aligned}
\]

where $\mathbf{1}[P]$ is $1$ if the statement $P$ is true and $0$ otherwise. Notice that a coefficient may be $2$ if the two edges of a length-$2$ leg have the same type.

Together with non-negativity and integrality of all variables, these constraints form an integer linear program.

We now prove correctness. Suppose first that the \textsc{FMMP} instance has a feasible solution whose body has color $s$. For each precoloring pattern, set the corresponding variables equal to the number of legs receiving each possible final coloring. The pattern constraints are then satisfied by construction. Moreover, the edge-type constraints hold because they simply count, for each type $\{i,j\}$, the number of graph edges of that type in the feasible coloring. Hence the integer program is feasible.

Conversely, suppose the integer program has a feasible integer solution. For every precoloring pattern, assign the indicated final colors to exactly the number of legs specified by the corresponding variables. Legs with the same precoloring pattern are interchangeable, so this can always be done. The constraints defining the variables ensure that all initial precolors are respected.

The edge-type constraints guarantee that, for every $\{i,j\}$, the number of graph edges whose endpoints receive colors $i$ and $j$ is exactly $s_{ij}$. By Observation~\ref{obs:coloring-formulation}, the resulting coloring therefore defines a feasible \textsc{FMMP} solution.

There are at most $(c+1)c$ variables of the form $x_{p,a}$ and at most $(c+1)^2c^2$ variables of the form $y_{p,q,a,d}$. Thus the total number of variables is $(c+1)c+(c+1)^2c^2=O(c^4)$. The number of constraints is also bounded by a function of $c$.

Integer linear programming with a fixed number of variables can be solved in polynomial time~\cite{Lenstra1983}. Since the dimension of our integer program depends only on $c$, it can be solved in time $f(c)\cdot |I|^{O(1)}$ for some computable function $f$. Trying all possible colors of an initially uncolored body introduces only an additional factor of $c$. Hence \textsc{FMMP} is fixed-parameter tractable parameterized by $c$ on spiders whose legs have length at most $2$.
\end{proof}

Theorem~\ref{thm:short-spider-fpt} yields a sharp contrast with Theorem~\ref{thm:spider-w1-hard}: allowing arbitrarily long legs makes the problem W[1]-hard parameterized by $c$, whereas restricting all legs to length at most $2$ yields fixed-parameter tractability even with arbitrary precoloring. The classical complexity of the latter case when $c$ is unbounded remains open.

\section{Hardness for Spiders with Legs of Length Three}
\label{sec:length-three}

The polynomial-time algorithm of Iburi and Uehara for completely uncolored spiders whose legs have length at most two~\cite{IburiUehara2024}, together with Theorem~\ref{thm:short-spider-fpt}, suggests that the maximum leg length plays an important role in the complexity of \textsc{FMMP}. We now consider the first larger value and show that, when precoloring is allowed, requiring every leg to have length exactly three is already sufficient for NP-completeness.

We reduce from \textsc{Partial Orientation}. An instance consists of a graph $H=(V,E)$ together with three nonnegative integers $\rho(v)$, $\delta(v)$, and $\theta(v)$ for every $v\in V$, satisfying $\rho(v)+\delta(v)+\theta(v)=\deg_H(v)$. The objective is to decide whether every edge can be either oriented in one of its two directions or left undirected so that every vertex $v$ has indegree $\rho(v)$, outdegree $\delta(v)$, and undirected degree $\theta(v)$. Pálvölgyi proved that \textsc{Partial Orientation} is NP-complete, even on planar graphs~\cite{Palvolgyi2009}.

\begin{theorem}
\label{thm:length-three}
\textsc{FMMP} is NP-complete on spiders in which every leg has length exactly $3$, when precoloring is allowed.
\end{theorem}

\begin{proof}
Membership in NP is immediate. For NP-hardness, let $H=(V,E)$ together with the functions $\rho$, $\delta$, and $\theta$ be an instance of \textsc{Partial Orientation}. We may assume that $\displaystyle\sum_{v\in V}\delta(v)=\displaystyle\sum_{v\in V}\rho(v)$, since otherwise the instance is immediately infeasible. Let $k=\displaystyle\sum_{v\in V}\delta(v)=\displaystyle\sum_{v\in V}\rho(v)$ and let $2h=\displaystyle\sum_{v\in V}\theta(v)$. Since $\rho(v)+\delta(v)+\theta(v)=\deg_H(v)$ for every $v$, we have $|E|=k+h$.

For every vertex $v\in V$, introduce a distinct color $x_v$. We additionally use four colors $s,t,p,q$. The body of the constructed spider is precolored with $s$.

The set of sticks is defined as follows. For every edge $uv\in E$, we introduce one stick of type ${x_u,x_v}$. For every vertex $v\in V$, we introduce $\delta(v)$ sticks of type $\{s,x_v\}$, $\rho(v)$ sticks of type ${x_v,t}$, and $\theta(v)$ sticks of type $\{q,x_v\}$. Finally, we introduce $h$ sticks of each of the types $\{s,q\}$, $\{s,p\}$, and $\{p,q\}$.

The total number of sticks is
\((k+h)+k+k+2h+h+h+h=3k+6h.\)

We construct a spider with $k+2h$ legs, all of length exactly $3$. The legs are divided into three types.

First, we introduce $k$ \emph{directed legs}. On each such leg, only the leaf is precolored, with color $t$.

Second, we introduce $h$ \emph{undirected legs}. On each such leg, the first vertex after the body is precolored with color $q$, while the other two vertices are uncolored.

Third, we introduce $h$ \emph{absorber legs}. On each absorber leg, the first two vertices after the body are precolored with colors $p$ and $q$, respectively, while the leaf is uncolored.

Thus the spider has $3(k+2h)=3k+6h$ edges, exactly as many as there are sticks.

We prove that the constructed \textsc{FMMP} instance is feasible if and only if the \textsc{Partial Orientation} instance is feasible.

Suppose first that $H$ admits a partial orientation satisfying all prescribed degrees. Consider a directed edge $u\rightarrow v$. We associate it with one directed leg and color the two internal vertices of that leg with $x_u$ and $x_v$, in this order. We place sticks of types ${s,x_u}$, ${x_u,x_v}$, and ${x_v,t}$ on the three edges of the leg.

There are exactly $\delta(v)$ directed edges leaving each vertex $v$ and exactly $\rho(v)$ directed edges entering it. Hence all sticks of types $\{s,x_v\}$ and ${x_v,t}$ are used exactly once.

Now consider an undirected edge $uv$. Choose an arbitrary ordering of its endpoints, say $(u,v)$. We associate this edge with one undirected leg and one absorber leg. On the undirected leg, use the colors $s,q,x_u,x_v$ and place sticks of types $\{s,q\}$, ${q,x_u}$, and ${x_u,x_v}$. On the corresponding absorber leg, use the colors $s,p,q,x_v$ and place sticks of types $\{s,p\}$, $\{p,q\}$, and $\{q,x_v\}$.

Every undirected edge therefore consumes one stick $\{q,x_v\}$ for each of its two endpoints. Since the undirected degree of vertex $v$ is $\theta(v)$, all $\theta(v)$ sticks of type $\{q,x_v\}$ are used. There are exactly $h$ undirected edges, so all $h$ sticks of types $\{s,q\}$, $\{s,p\}$, and $\{p,q\}$ are also used. Hence every stick is placed and the constructed instance is feasible.

Conversely, suppose that the constructed \textsc{FMMP} instance is feasible. Consider first the absorber legs. Their first edge joins vertices precolored $s$ and $p$, so it must receive a stick of type $\{s,p\}$. Similarly, their second edge must receive a stick of type $\{p,q\}$. Since there are exactly $h$ absorber legs and exactly $h$ sticks of each of these types, all such sticks are consumed by absorber legs.

Likewise, the first edge of every undirected leg joins vertices precolored $s$ and $q$, and must therefore receive a stick of type $\{s,q\}$. Hence all $h$ sticks of type $\{s,q\}$ are consumed by the undirected legs.

After these placements, the only unused sticks containing color $s$ are the $k$ sticks of types $\{s,x_v\}$. There are exactly $k$ remaining edges incident with the body, namely the first edges of the directed legs. Therefore every directed leg begins with some stick ${s,x_u}$.

Similarly, the only sticks containing color $t$ are the $k$ sticks of types ${x_v,t}$, and the only vertices precolored $t$ are the leaves of the $k$ directed legs. Thus the last edge of every directed leg receives some stick ${x_v,t}$.

Consequently, the middle edge of every directed leg joins two vertices colored $x_u$ and $x_v$. It must therefore receive one of the sticks corresponding to an edge $uv\in E$. We orient this source edge from $u$ to $v$.

Exactly $k$ source-edge sticks are used in this way. Since there are $k+h$ source-edge sticks in total, exactly $h$ remain.

All sticks of types $\{s,p\}$, $\{p,q\}$, and $\{s,q\}$ have already been consumed. The remaining sticks incident with color $q$ are therefore precisely the $2h$ sticks of types $\{q,x_v\}$. There are exactly $2h$ still-unfilled edges incident with a vertex precolored $q$: the second edges of the $h$ undirected legs and the final edges of the $h$ absorber legs. Hence all sticks of type $\{q,x_v\}$ are used on these positions.

It follows that the final edge of every undirected leg must receive one of the $h$ remaining source-edge sticks. We declare the corresponding source edges to be undirected.

We have now assigned every edge of $H$ either a direction or the undirected status. It remains to verify the prescribed degrees. A directed edge leaves $v$ exactly when its directed leg begins with a stick of type $\{s,x_v\}$. Since there are exactly $\delta(v)$ such sticks, the outdegree of $v$ is $\delta(v)$. Similarly, a directed edge enters $v$ exactly when its directed leg ends with a stick of type ${x_v,t}$, so its indegree is $\rho(v)$.

Finally, every edge incident with $v$ that is not directed is undirected. Therefore its undirected degree is
$\deg_H(v)-\delta(v)-\rho(v)=\theta(v)$.

Thus the resulting partial orientation satisfies all prescribed degrees. The original \textsc{Partial Orientation} instance is feasible if and only if the constructed \textsc{FMMP} instance is feasible, completing the reduction.
\end{proof}

Theorem~\ref{thm:length-three} shows that the first leg length beyond the previously known tractable uncolored case already permits NP-completeness when precoloring is available. Whether precoloring is necessary for hardness at this exact leg length remains open.

\section{Conclusion and Open Problems}
\label{sec:conclusion}

We studied how the number of colors and the structure of spiders affect the complexity of \textsc{Matching-Match}. We showed that two colors already suffice for NP-completeness on general graphs, even without precolored vertices and with maximum degree three. On spiders, \textsc{FMMP} is W[1]-hard parameterized by $c$, even when only the body is precolored.

Bounding the leg length changes this picture. For spiders whose legs have length at most two, we gave an FPT algorithm parameterized by $c$ under arbitrary precoloring, while spiders whose legs all have length exactly three are already NP-complete when precoloring is allowed. Two main questions remain open. First, is \textsc{FMMP} polynomial-time solvable or NP-complete on arbitrarily precolored spiders with legs of length at most two when $c$ is part of the input. Second, what is the complexity of completely uncolored spiders whose legs have length exactly three~\cite{IburiUehara2024}. From the parameterized perspective, it is also natural to determine whether the latter case is FPT or W[1]-hard parameterized by $c$.

\bibliographystyle{splncs04}
\bibliography{bibl}

\end{document}